\documentclass[submission,copyright,creativecommons]{eptcs}
\providecommand{\event}{LSFA 2024} 

\usepackage{iftex}

\ifpdf
  \usepackage{underscore}         
  \usepackage[T1]{fontenc}        
\else
  \usepackage{breakurl}           
\fi

\usepackage[english]{babel}
\usepackage[utf8]{inputenc}
\usepackage{graphicx}

\usepackage{url}

\usepackage{amssymb}
\usepackage{amsthm}
\usepackage{amsmath}
\usepackage{latexsym}
\usepackage{mathtools}
\usepackage{stmaryrd}

\usepackage{quiver}
\usepackage{tikz}
\usepackage{tikz-cd}
\tikzset{
  no line/.style={draw=none,
    commutative diagrams/every label/.append style={/tikz/auto=false}},
  from/.style args={#1 to #2}{to path={(#1)--(#2)\tikztonodes}}}

\usepackage{iohk}
\usepackage{extarrows}
\usepackage{slashed}
\usepackage{xcolor}
\usepackage{float}
\usepackage{microtype}
\usepackage{hyperref}

\theoremstyle{definition}
\newtheorem{theorem}{Theorem}[subsection]
\newtheorem{lemma}[theorem]{Lemma}
\newtheorem{proposition}[theorem]{Proposition}
\newtheorem{algorithm}[theorem]{Algorithm}
\newtheorem{corollary}[theorem]{Corollary}

\newtheorem{definition}[theorem]{Definition}
\newtheorem{example}[theorem]{Example}
\newtheorem{remark}[theorem]{Remark}

\theoremstyle{plain}

\usepackage{array}
\usepackage{subcaption}
\usepackage{caption}
\usepackage{subfiles}
\usepackage{listings}
\usepackage{verbatim}
\usepackage{listings}
\usepackage{alltt}
\usepackage{paralist}

\usepackage{semantic}

\newcommand{\red}[1]{\textcolor{red}{#1}}

\makeatletter
\DeclareTextCommand{\textprime}{\encodingdefault}{%
  \mbox{$\m@th'\kern-\scriptspace$}%
}
\makeatother

\newcommand{\s}{\textsf}  

\newcommand{\true}{\type{True}}

\newcommand{\List}[1]{\ensuremath{\s{List}[#1]}}

\newcommand{\emptytype}{\star}

\newcommand{\inputs}{\fun{inputs}}
\newcommand{\outputs}{\fun{outs}}

\newcommand{\Input}{\ensuremath{\s{Input}}}
\newcommand{\Output}{\type{Output}}
\newcommand{\OutputRef}{\fun{OutputRef}}

\newcommand{\outputref}{\fun{outputRef}}

\newcommand\N{\ensuremath{\mathbb{N}}}

\newcommand{\leteq}{\ensuremath{\mathrel{\mathop:}=}}

\newcommand{\Bool}{\type{Bool}}

\newcommand{\STRUC}{\type{STRUC}}

\newcommand{\LEDGER}{\type{LEDGER}}

\newcommand{\State}{\type{State}}

\newcommand{\Tx}{\type{Tx}}

\newcommand{\UTxO}{\type{UTxO}}

\newcommand{\True}{\type{True}}

\newcommand{\Ix}{\type{Ix}}
\newcommand{\Slot}{\type{Slot}}

\newcommand{\Env}{\type{Env}}

\usepackage{etoolbox}
\usepackage{tikz-qtree}
\usepackage{tikz}
\usetikzlibrary{matrix}

\newtoggle{anonymous}
\toggletrue{anonymous}
\iftoggle{anonymous}{

}{

}

\newcommand{\TRANS}{\type{TRANS}}

\newcommand{\Trc}{\type{Trc}}

\newcommand{\blank}{
    \hspace{0.04cm}
    \rule{2.4mm}{.4pt}
    \hspace{0.04cm} }
\newcommand{\blankd}{
    \hspace{0.04cm}
    \rule{1.5mm}{.4pt}
    \hspace{0.04cm} }

\mathchardef\mhyphen="2D

\newcommand{\ByteString}{\type{ByteString}}

\newcommand{\Def}{\operatorname{Def}}
\newcommand{\GraphSieve}{\normalfont\textbf{Graph}^{\sharp}}
\newcommand{\SetCategory}{\normalfont\textbf{Set}}
\newcommand{\UMet}{\normalfont\textbf{UMet}}

\newcommand{\ListInf}[1]{\ensuremath{\s{List}_{\infty}[#1]}}

\newcommand{\comma}{,}

\renewcommand{\Im}{\operatorname{Im}}

\makeatletter
\DeclareFontFamily{OMX}{MnSymbolE}{}
\DeclareSymbolFont{MnLargeSymbols}{OMX}{MnSymbolE}{m}{n}
\SetSymbolFont{MnLargeSymbols}{bold}{OMX}{MnSymbolE}{b}{n}
\DeclareFontShape{OMX}{MnSymbolE}{m}{n}{
    <-6>  MnSymbolE5
   <6-7>  MnSymbolE6
   <7-8>  MnSymbolE7
   <8-9>  MnSymbolE8
   <9-10> MnSymbolE9
  <10-12> MnSymbolE10
  <12->   MnSymbolE12
}{}
\DeclareFontShape{OMX}{MnSymbolE}{b}{n}{
    <-6>  MnSymbolE-Bold5
   <6-7>  MnSymbolE-Bold6
   <7-8>  MnSymbolE-Bold7
   <8-9>  MnSymbolE-Bold8
   <9-10> MnSymbolE-Bold9
  <10-12> MnSymbolE-Bold10
  <12->   MnSymbolE-Bold12
}{}

\let\llangle\@undefined
\let\rrangle\@undefined
\DeclareMathDelimiter{\llangle}{\mathopen}%
                     {MnLargeSymbols}{'164}{MnLargeSymbols}{'164}
\DeclareMathDelimiter{\rrangle}{\mathclose}%
                     {MnLargeSymbols}{'171}{MnLargeSymbols}{'171}
\makeatother

\newcommand{\listsinf}[1]{\left\llangle #1 \right\rrangle}

\title{Properties of UTxO Ledgers and Programs Implemented on Them}
\author{Polina Vinogradova
\institute{Input Output Global}
\email{polina.vinogradova@iohk.io}
\and
Alexey Sorokin
\institute{Input Output Global}
\email{alex.sorokin@iohk.io}
}
\newcommand\titlerunning{Properties of UTxO Ledgers and Programs Implemented on Them}
\newcommand\authorrunning{P. Vinogradova, A. Sorokin}

\hypersetup{
  bookmarksnumbered,
  pdftitle    = {\titlerunning},
  pdfauthor   = {\authorrunning},
  pdfkeywords = {blockchain , ledger , smart contract , formal verification , specification , transition system , UTxO , properties , safety} 
}

\begin{document}
\maketitle

\begin{abstract}
    Trace-based properties are the gold standard for program behaviour analysis. One of the domains of application of this type of analysis is cryptocurrency ledgers, both for the purpose of analyzing the behaviour of the ledger itself, and any user-defined programs called by it, known as smart contracts. The (extended) UTxO ledger model is a kind of ledger model where all smart contract code is stateless, and additional work must be done to model stateful programs. We formalize the application of trace-based analysis to UTxO ledgers and contracts, expressing it in the languages of topology, as well as graph and category theory. To describe valid traces of UTxO ledger executions, and their relation to the behaviour of stateful programs implemented on the ledger, we define a category of simple graphs, infinite paths in which form an ultra-metric space. Maps in this category are arbitrary partial sieve-define homomorphisms of simple graphs. Programs implemented on the ledger correspond to non-expanding maps out of the graph of valid UTxO execution traces. We reason about safety properties in this framework, and prove properties of valid UTxO ledger traces. \newline

\end{abstract}
    
\section{Introduction}
\label{sec:intro}

Cryptocurrency ledgers are distributed ledgers keeping records of
digital currency. When a user submits a transaction to
the network, each local record state is updated by applying the changes specified in the transaction.
Cryptocurrency ledger behavior is often described by as a deterministic
state transition systems \cite{tezos,ethereum,bitcoin,nervos,zil}. 
The main functionality supported by blockchain ledgers programs is a single atomic operation:
the application of a transaction to the given state, and all state updates can be decomposed into 
applications of single transactions. This makes small-step semantics 
well-suited to specify ledger behavior, as exemplified in both scientific research
and in industrial implementations \cite{structured,alonzo}. 

The two most common ledger models are \emph{account-based} (such as Ethereum \cite{ethereum}), and 
\emph{UTxO-based}. A minimal account-based ledger records the state of a collection of accounts, 
and applying transactions to the ledger updates the account states, and transfers funds between accounts. 
In this work, we focus on the UTxO ledger model, and its formalization in terms 
of small-step semantic \cite{structured}. 
A UTxO ledger, which stands for \emph{Unspent Transaction Outputs}, records entries (i.e. 
\emph{transaction outputs}) in a \emph{UTxO set}. 
This model was first introduced in the BitCoin system \cite{bitcoin},
and currently in use by the Cardano \cite{alonzo} and Ergo
\cite{ergo}. A UTxO entry contains some funds, 
together with a specification of what kinds of transactions are allowed to "spend" this record,
e.g. only ones signed by a particular key. Applying a transaction to a
UTxO set only adds of removes entries, never modifying them. This makes 
performing formal analysis on this model more tractable in many cases. In particular, 
the outcome of a transaction application is not affected by the order of transaction 
application, as we demonstrate in this work. 

An \emph{extended} UTxO (EUTxO) ledger is a UTxO-style ledger that supports the use of \emph{smart contracts}.
EUTxO smart contracts are pieces of user-defined code that specify permissions for
certain transaction actions. In this work, we omit most of the details of this
mechanism as it is not necessary for our construction.
However, we do make use of a stateful contract computation model
called a \emph{structured contract framework} (SCF) \cite{structured} that relies on the possibility of
composing sophisticated permissions. This framework specifies the condition 
that any valid transaction applied to the ledger state will result 
in an update for the program state encoded in that ledger state
that is in accordance with the program specification. This condition is used as the definition
of a a program being implemented on the ledger, including both consolidated and distributed programs.
This approach only formalizes a single "correct" step of a program, both in its ledger implementation 
and specification.

Trace-based properties \cite{liveness} are predicates on infinite sequences of program
states used to reason about the correctness of stateful programs. Unlike the structured contract 
formalism, this approach allows for reasoning about multiple program steps at a time, 
as it studies program executions of infinite length. 
As is, in order to demonstrate the correctness of the (single step) program state update,
the SC approach requires making certain assumptions, which  
are related to the behavior of the UTxO ledger program itself. These assumptions 
can only be expressed in terms of subsets of ledger traces generated in a specific way, i.e. 
a specific trace-based property. Moreover, the basic definitions of
the SCF formalism have not been used to to reason about traces or express any sophisticated guarantees 
of correct behavior. 

We make the observation that correctly-implemented programs, then, can be shown to behave correctly only when 
their behavior is considered on ledger traces that are themselves generated from a valid 
start state, and according to the ledger program specification.
In this work, we take a principled approach to formalizing this idea, while filling these gaps 
in analysis of the ledger program
and the implementations of other programs on it. In particular, we
use ideas from trace-based behavior analysis, graph theory, category theory, and topology, to 
do the following analysis:

\begin{itemize}
  \item[(i)] we define what it means for a ledger or contract trace to have the property of being \emph{valid},
  defined in terms of small-steps specifications and initial states;

  \item[(ii)] we define a category of simple graphs with fixed initial vertices,
  paths in which represent valid execution traces,
  and the maps are partial sieve-defined homomorphisms between them;

  \item[(iii)] we demonstrate that infinite paths in the graphs of this category form an ultrametric space \cite{ultrametric}.
  As a corollary, we obtain that any structured contract induces a non-expanding (and therefore continuous)
  map from ledger traces to structured contract traces;

  \item[(v)] we prove a number of classic UTxO ledger safety properties within our framework , including
  transaction commutativity for valid permutations, and replay protection.
\end{itemize}

The value in defining the category in (ii) is that it formalizes the relation of being "correctly implemented on" between 
\emph{a stateful program} and the \emph{ledger program which executes it} in a way that aligns with existing trace-based behavior 
analysis techniques. Formalizing the relation between properties of 
these programs allows us to derive useful results about ledger behavior based on the properties of implemented programs, 
and vice versa. In particular, a direct consequence of the nature of our construction is that a program property of the 
form "a specific bad thing never happens" (i.e., a safety property) is satisfied by all valid execution traces 
generated from a correct implementation of that program on the 
ledger. Finally, we note that the results in (v) are proved, instead of being treated as an assumption as is done in some existing 
work on implementing programs on a UTxO ledger in order to prove the correctness of its implementation \cite{msgs}. Since 
these trace properties cannot be expressed as predicates on an a specific state, this required the additional machinery 
we introduce here to formalize.

\section{Small-step specifications and the ledger model}
\label{sec:sts}

We present an overview of the types and semantics employed in our specifications,
on which we later base our analysis.
\subsection{Small-step specifications}

Small-step semantics are specified in terms of of atomic steps from which all other (composite) steps can be built. We establish the following notation for our upcoming specifications: 

\begin{definition}
  A \emph{small-step transition system} $\mathsf{TRANS}$ is given by a subset $\mathsf{TRANS} \subseteq \Env \times \State \times \Input \times \State$, 
  also denoted 
    $
    e \vdash
    \var{s} \trans{trans}{i} \var{s'} ~
  $
\end{definition}

Each component in a quadruple $(e,~s,~i,~s')$ has the following role: (i) $e \in \Env$ is an \emph{environment}; (ii) $s \in \State$ is a state to which an \emph{input} $i\in\Input$ is applied, or a \emph{starting state}; (iii) $s' \in \State$ is a state representing, for the given environment, a result of the input application to the starting state, or the \emph{end state}. It is possible that a given state has no valid transitions out of it, or that a given input is not included in valid transition 4-tuple.

Given a transition $(e,~s,~i,~s') \in \TRANS$, the pair $(e,~i)$ of an environment and an input make up a \emph{transition label} of $\TRANS.$ The difference between the input and the environment is in the source of the data. This convention comes from the specification of the Cardano ledger transition system \cite{shelley}. The user issues the input, e.g., a transaction. The environment is sourced from the transaction-containing blocks (which we do not model here). For example, the consensus and block production mechanisms keep track of the current time, which is specified by the environment in our model.

\subsection{Ledger transition system}

A ledger is a stateful program that implements one operation: the application of a transaction to the current ledger state. In the Extended UTxO (EUTxO) ledger model, the state is made up of a set of \emph{unspent transaction outputs}, called the UTxO set. A single UTxO entry (unspent output) consists of a unique identifier together with an output. The output contain data, assets, and code (i.e., a smart contract, which specifies what transactions have are permitted to remove this output from the UTxO set). UTxO entries are immutable, can can only be either added or removed from the UTxO set. 
We give an overview of the EUTxO state and transaction structure below, leaving abstract any types and definitions that are not relevant for the examples we present later. A complete specification is available in \cite{structured}. Nonstandard notation we use here is 
specified in Figure \ref{fig:notation:nonstandard}.

\begin{figure}[htb]
  \begin{align*}
    \emptytype
    & :~\{\emptytype\}
    & \text{denotes the one-element set, and its one inhabitant}
    \\
    Q
    & :~\type{Set}~(A)
    & \text{$Q$ is a set of elements of type $A$}
    \\
    \var{Key} \mapsto \var{Value}
    & \subsetneq [(\var{Key},~ \var{Value})]
    & \text{finite map}
  \end{align*}
  \caption{Non-standard map operators}
  \label{fig:notation:nonstandard}
\end{figure}

\textbf{UTxO set.} A \emph{UTxO set} constitutes the state of a ledger model. It is given by a finite map (finite associative array)
$ 
  (\ByteString, \N) \mapsto \Output
$
The key in the $\UTxO$ finite map is called an \emph{output reference}. 

\textbf{Slot number.} To represent blockchain time, a natural number is used, which we call a \emph{slot number} (or just \emph{slot}), with $\Slot = \N.$ A pair of slot numbers represents an interval which includes the first (starting) slot of the pair, but excludes the second (end) slot.

\textbf{Transaction.} Updates that a user wants to make to the UTxO set are specified in a $\Tx$ and $\Input$ data structures
\begin{align*}
  \Tx &=(\inputs: \type{Set}~{(\Input)}, & \Input &=( \outputref: (\ByteString, \N),\\
           & \ \outputs: \List{\Output}, & & \ \fun{output}: \Output, \\
           & \ \fun{validityInterval}: ({\Slot},{\Slot}), \\
           & \ \fun{additionalData} : \fun{AdditionalData})
\end{align*}

A transaction $\var{tx}$ consists of (i) a set $\inputs (\var{tx})$ of inputs (pointers to entries in the UTxO set), (ii) 
outputs $\outputs (\var{tx})$ (which will be added as values to the UTxO set, accompanied by appropriate unique identifier pointers), (iii) 
a pair of slots $\fun{validityInterval} (\var{tx})$ representing the interval of time during which this transaction can be 
processed.
We include the field $\fun{additionalData}$ in order to suggest how to extend this model to represent the full EUTxO ledger model \cite{structured}. Also, $\Tx$ is equipped with a hash function 
$
  h: \Tx \to \ByteString.  
$
that is used, for a given transaction $\var{tx}$, to generate the unique identifier $(\fun{hash}~\var{tx}, \var{ix})$. This identifier, called an output reference, is used when adding output $o_{ix} \in \outputs~\var{tx}$ with index $\var{ix}$ (in the list of outputs of $\var{tx}$) to the UTxO set. 

The following gives the type of the ledger transition system $\mathsf{LEDGER}
      \subseteq 
    \Slot \times \UTxO \times \Tx \times \UTxO
$
We require any member $(q,u,t,u')$ of $\LEDGER$ to satisfy the following constraints: (i) a transaction $t$ has at least one input; (ii) a slot $q$ is within the validity interval of transaction $t$; (iii) all transaction inputs exist in the UTxO set $u$. We define a function which checks these constraints,
$
  \fun{checkTx}:~\Slot \times \UTxO \times \Tx~\to~\Bool  
$ 
by the following rule
\[
  \fun{checkTx}~(q, u, t)~ \leteq (\inputs~\var{t}~\neq~\varnothing)  
  \wedge (q \in \fun{validityInterval}(\var{t}))  
\]
\[
  \wedge (\forall i \in \inputs~\var{t},~ (\outputref~i) \mapsto (\fun{output}~i)~\in~u) \wedge (\fun{additionalChecks}~(q, u, t))
\]
Again, the field $\fun{additionalChecks}$ is included in order to ensure the model can be extended to full EUTxO. The function $\fun{checkTx}$ specifies, for a given triple $(q,~u,~t)$, the conditions that must be satisfied in order for there to exist a $u'$ with $(q,u,t,u') \in \mathsf{LEDGER}.$ When it exists, $u'\in \UTxO$ is computed as follows: (i) UTxO entries in $u$ that correspond to the inputs of the transaction $t$ must be removed; (ii) entries constructed from the outputs of the transaction $t$ must be added to $u.$ We call the functions $\fun{getORefs}$ and $\fun{mkOuts}$ respectively (defined in Figure \ref{fig:utxo-func}). The following rule $\fun{ApplyTx}$ defines the membership of $(q,u,t,u')$ in the $\mathsf{LEDGER}$ relation:   
  \begin{equation}
    \label{fig:ledger-rule}
      \inference[ApplyTx]
      {
      u'~\leteq~(u \setminus \fun{getORefs}(t)) \cup \fun{mkOuts}(t)
      \\ ~ \\
      \fun{checkTx}~(q,u,t)
      \\ ~ \\
      }
      {
      \begin{array}{l}
        q \\
      \end{array}
        \vdash
        \begin{array}{r}
          u \\
        \end{array}
        \trans{ledger}{t}
        \begin{array}{r}
          u'  \\
        \end{array}
      }
  \end{equation}

  \begin{figure}[htb]
    \begin{align*}
      \fun{toMap}~ &: \Ix \to [\Output] \to (\Ix \mapsto \Output) \\
      \fun{toMap}~\var{ix}~\{\} &= [~] \\
      \fun{toMap}~\var{ix}~[u;~\var{outs}] &= \{~\var{ix}\mapsto u~\} \cup \{~(\fun{toMap}~(\var{ix}+1)~\var{outs})~\}
      \nextdef
      \fun{mkOuts}~ &: \Tx \to \UTxO \\
      \fun{mkOuts}~{tx} &= \{~(\var{tx},~\var{ix}) \mapsto o~ \mid~(\var{ix} \mapsto o)\in~\fun{toMap}~0~(\outputs~\var{tx})~\}
      \nextdef
      \fun{getORefs}~&~: \Tx \to \type{Set}~({\OutputRef}) \\
      \fun{getORefs}~{tx} &= \{~\outputref~i~\mid~i~\in~\inputs~\var{tx} ~\} 
    \end{align*}
    \caption{Auxiliary UTxO functions}
    \label{fig:utxo-func}
  \end{figure}  
\subsection{Structured contracts}
\label{sec:struc}

To motivate the category we define later, we specify of what it means for a \emph{stateful contract} to be implemented on the UTxO ledger. User-defined code in the EUTxO model is not stateful, and can only take the form of boolean predicates on transaction data. To reason about stateful contracts, we use the structured contract model \cite{structured}. Let $\STRUC$ be a program expressed in terms of the small-steps semantics. The state of this contract is represented (encoded) on the ledger in some specific way. This encoding is specified in terms of a (partial) projection function that computes the contract state for a given ledger state (or fails). Similarly, for a given transaction, the input to $\STRUC$ can be computed. The structured contract specification $\STRUC$ is said to be \emph{implemented correctly} whenever, for a given ledger state and transaction, the structured contract state (encoded within that ledger state) is updated by the input (encoded withing the transaction) in accordance with the $\STRUC$ specification. 

The specification $\STRUC$, together with the projection functions and a proof the the correctness of the implementation, constitutes a structured contract. 
Note here that we do not show the exact details of running user-defined code predicates on transaction data. In a full model, we would fill in these details by specifying the required $\fun{additionalData}$ and $\fun{additionalChecks}$. Depending on the choice of projection functions, this model can be used to represent both distributed and consolidated contracts \cite{structured}. We formalize,

\begin{definition}\label{def:structured}

  Suppose $\STRUC$ is small-step transition system 
$
    \STRUC
      \subseteq 
    \{*\} \times \State \times \Input \times \State
$
  and let we have a partial function $\pi : \UTxO \rightharpoonup \State$ and a (total) function \footnote{A total function $f$ from $X$ to $Y$ means everywhere-defined function on $X,$ while a partial function $f$ from $X$ to $Y$ means a total function from a subset $\Def(f)\subseteq X$ to $Y.$} $\kappa : \Tx \to \Input$ such that
    \begin{equation}
      \inference[$~$] 
      {
        \\~\\
        (u \in \Def\pi) ~~ \land ~~
        \left(
        {
          \begin{array}{c}
            q\\
          \end{array}
        }
        \vdash
        {
            \begin{array}{r}
              \var{u} \\
            \end{array}
        }
        \trans{ledger}{\var{t}}
        {
            \begin{array}{r}
              \var{u'} \\
            \end{array}
        }
        \right)
        \\~\\
      }
      {
        \\ ~ \\
        (u' \in \Def\pi)  ~~ \land ~~
        \left(
        {
          \begin{array}{c}
            *\\
          \end{array}
        }
        \vdash
        {
            \begin{array}{r}
              \pi\var{u} \\
            \end{array}
        }
        \trans{struc}{\kappa\var{t}}
        {
            \begin{array}{r}
              \pi\var{u'} \\
            \end{array}
        }
        \right)
        \\~\\
      }
    \end{equation}
  The triple $(\STRUC, \pi, \kappa)$ is called a \emph{structured contract} for $\LEDGER.$
\end{definition}

Following existing EUTxO design, no block-level data is exposed to a smart contract, such as a current slot number, so a singleton $\{*\}$ is the type of a structured contract specification environment. An example of a very basic program that can be defined as a structured contract would be an NFT, or a non-fungible (i.e. unique) token \cite{structured}. The state of an NFT contract is the total quantity of this type of token in the UTxO set, and the constraint in the update rule of that state is that the updated (end state) quantity may not exceed 1. 

As is, structured contracts give only the guarantee that each program step will be correct, both in the specification and in the ledger implementation. In the rest of this work, define the space of valid program executions constructed from its small-step specification in a way that aligns with existing definitions and terminology in the field. This allows for expressing (and proving adherence to) a wider range of properties, and formalizing the relationship between valid executions of ledger programs and executions of their implementations. 

\section{Traces}

In \cite{liveness}, the collection of \emph{traces} (or \emph{executions}) of a program with states of type $S$ is a collection of infinite lists of states. Taking $S = \UTxO$, we get a set of all UTxO traces, and taking $S = \State,$ we get the set of all contract state traces. In both cases we are dealing with \emph{arbitrary lists of states}. However, we are interested in reasoning about traces that are generated according to their specification. Therefore, we need to define appropriate subsets of infinite lists, which we refer to as \emph{valid traces}:

\begin{definition}[Valid ledger and structured contract traces] ~

  \begin{itemize}
    \item[(i)] Fix subsets 
    $
      \UTxO_0 \subseteq \UTxO, \Slot_0\subseteq\Slot
    $ 
    called \emph{valid initial UTxO states} and \emph{valid initial slots}, respectively. A \emph{set of valid $\LEDGER$ traces} is 
    \begin{align*}
      \Trc({\LEDGER})~ \leteq~& 
      \{ (u_0, u_1,\ldots) \in \ListInf{\UTxO}~\mid~ 
      \forall j \geq 0,~\exists~(q_j, u_j, t_j, u_{j+1}) \in\LEDGER:
      \\
      ~ & 
      (q_0,u_0) \in \Slot_0\times\UTxO_0,~ q_0\leq q_1\leq \ldots~\}
    \end{align*}
    Note that it is not enough to require that each 4-tuple is in $\LEDGER$, the slots must appear in the increasing order.

    \item[(ii)] Let $(\STRUC,~\pi,~\kappa)$ be a structured contract for $\LEDGER.$ We say that 
    $
      \State_0 \subseteq \State
    $ 
    is a set of \emph{valid initial $\STRUC$ states} if all valid initial UTxO states are mapped to $\State_0$ via the partially defined map $\pi,$ i.e. 
    $
      \UTxO_0 \subseteq \Def \pi$ and $\pi(\UTxO_0) \subseteq \State_0.
    $
    A \emph{set of valid $\STRUC$ traces} is 
    \[
      \Trc({\STRUC})~ \leteq~ \{~(s_0, s_1,\ldots) \in \ListInf{\State}~\mid~
      \forall j\geq 0, ~(*,~s_j,~i_j,~s_{j+1})~\in~\STRUC, ~s_0 \in \State_0~\}
    \]
\end{itemize}
\end{definition}

Since the environment type for structured contracts is $\{*\},$ there is no restriction on it. For both $\Trc({\LEDGER})$ and $\Trc({\STRUC})$ we refer to lists of quadruples corresponding to valid traces in $\LEDGER$ and $\STRUC,$ respectively, as \emph{lifts} of that traces.

\section{Graphs and sieve-defined homomorphisms}
\label{sec:graphs_sieves_hmms}

The specification $\LEDGER$ can be used to generate sequences of the form $(q,u,t_0,u_0),~(q,u_0,t_1,u_1),~...$, such that for each triple $(q,u_i,t_i), \fun{checkTx}~ (q,u_i,t) = \true$. This process gives rise to a \emph{simple graph}\footnote{By a simple graph we mean a directed graph in which for any two vertices $x$ and $y$ there is at most one edge $e:x\to y.$ Single loops at vertices are also allowed.} $\Lambda$, whose vertices are triples $(q,u,t)$ satisfying~$\fun{checkTx}$, and edges connecting any two $(q,u,t)$ and $(q',u',t')$ whenever $(q,u,t,u') \in \LEDGER$. Similarly, a structured contract specification $\STRUC$ defines another simple graph $\Gamma$ on triples $(*,s,i)$. We formalize the relation between simple graphs in general, and in particular, between graphs related via implementation, such as $\LEDGER$ and $\STRUC$. 

Note that in this section and onwards, we use standard terms and definitions of category theory, as described in existing works \cite{mac2013categories, borceux1994handbook, borceux1994handbook3}.
\subsection{Simple graphs and sieve-defined homomorphisms}

In general, a homomorphism of directed graphs is a pair of maps: one for vertices and one for edges, such that the map of edges respects their domains and codomains. In the case of a simple graphs, a map for vertices completely defines the homomorphism.

Let $G$ be the graph generated from $\LEDGER$, and $G'$ be the graph generated from $G'$.
The partially defined function $\pi: \UTxO \rightharpoonup \State$ and total function $\kappa:\Tx\to\Input$ define a partial map $\varphi$ from the set of vertices of $G$ to the set of vertices of $G',$ such that the domain of definition of $\varphi$, denoted $\Def \varphi$, is \emph{upward closed}. That is, for an edge $(q,u,t) \to (q',u',t')$ in $G$, and $(q,u,t)\in\Def \varphi$, then $(q',u',t')\in \Def \varphi$. This follows immediately from the definition of structured contract. Moreover, $\varphi$ defines a homomorphism of simple graphs from $\Def \varphi$ to $G'$:
$
    \left(
        (q,u,t) 
            \xrightarrow{~~}
        (q',u',t')
    \right)
    \mapsto
    \left(
        (*,\pi u,\kappa t) 
            \xrightarrow{~~}
        (*,\pi u',\kappa t')
    \right) .~
$

This can be made precise as the following definition,

\begin{definition}\label{def:Sieve}
    Let $G=(V,E)$ be a graph. A subset $S\subseteq V$ is called a \textit{sieve}, whenever any edge starting in $S$ necessary ends in $S$. That is, if $e:v_1\to v_2 \in E$ and $v_1\in S$, then $v_2 \in S$. Equivalently, every path in $G$ that starts in $S$ also ends in $S$.       
\end{definition}

The following are trivial examples of sieves:

\begin{example}
    For any graph $G=(V,E)$ the set of all vertices $V$ and any sink in $G$ are sieves. Also, any intersection of sieves is a sieve.
\end{example}

We can now define the class of graph homomorphisms that define the implementation relation between two specifications in graph-theoretic terms:

\begin{definition}\label{def:SieveDefHmm}
    Let $G=(V,E)$ and $G'=(V',E')$ be simple graphs. A partial map $\varphi:V\rightharpoonup~V',$ whose domain of definition $\Def \varphi$ is a sieve in $G,$ is called a \textit{partially sieve-defined homomorphism}, if $\varphi$ extends to a homomorphism of graphs from the full subgraph of $G$ on $\Def \varphi$ to the graph $G'.$ Such map will be denoted
    $
        \varphi: G \rightharpoonup G'    
    $ 
    and $\Def \varphi$ will also denote the corresponding full subgraph in $G$.
\end{definition}

All everywhere-defined homomorphisms of simple graphs are also partial sieve-defined homomorphism of simple graphs. In particular, the identity homomorphism is a trivial example of this. The following lemma is needed to define the composition of sieve-defined homomorphisms:

\begin{lemma}[Sieve of a sieve is a sieve]\label{def:CompositionOfSieves}
    Let $G=(V,E)$ be a graph and $S$ a sieve in $G.$ Suppose $S'$ is a sieve in a full subgraph $\overline{S}$ of $G$ defined by $S.$ Then $S'$ is a sieve in $G.$    
\end{lemma}

\begin{proof}
    Let $e:v \to v'$ be an edge in $G$ such that $v\in S'$. Since $S'\subseteq S$, then $v\in S$. Hence, $v'$ belongs to $S$, as $S$ is a sieve in $G$. Moreover, the edge $e$ is in the full subgraph $\overline{S}$ of $G$ on $S.$ By the assumption, $S'$ is a sieve in $\overline{S},$ so $v'$ belongs to $S'$.
\end{proof}

Recall the composition of partially defined functions between sets. Suppose 
$
    f:X \rightharpoonup Y 
    \text{~and~}
    g:Y \rightharpoonup Z
$
are partially defined functions. We define their composition to be the partially defined function 
$
    g\circ f: X \rightharpoonup Z, ~~
    \Def(g\circ f) = \Def f \cap f^{-1}(\Def g)
$
and $(g\circ f)x = g(f(x))$ if $x\in\Def(g\circ f).$

Since any partial sieve-defined homomorphism of simple graphs is completely defined by a map on its vertices, we define the composition of partial sieve-defined homomorphisms in a similar way to sets.

\begin{lemma}
    Partial sieve-defined homomorphisms of simple graphs are closed under composition.
\end{lemma}

\vspace{-4mm}

\begin{proof}
    Given 
$
    f:G \rightharpoonup G' 
    \text{~and~}
    g:G' \rightharpoonup G''
$     
it is enough to prove that  $\Def(g\circ f) = \Def f \cap f^{-1}(\Def g)$
is a sieve in $\Def f$. Suppose $e:v\to v'$ is an edge in $\Def f$ such that $v\in\Def(g\circ f)$. The edge $f(e):f(v)\to~f(v')$ then belongs to $G',$ where $f(v)\in\Def g.$ Since $\Def g$ is a sieve, $f(v')$ also belongs to $\Def g.$ Therefore, $v'\in\Def f \cap f^{-1}(\Def g)$. We then get that $\Def(g\circ f)$ is a sieve in $\Def f.$
Since a sieve of a sieve is a sieve, $g\circ f:G \rightharpoonup G''$ is a partial sieve-defined homomorphism of simple graphs.
\end{proof}

\vspace{-2mm}

The above information motivates us to introduce a category $\GraphSieve$ of all simple graphs and partial sieve-defined homomorphisms.
\subsection{Graphs with distinguished initial vertices}

The definition of the set of valid traces for $\LEDGER$ and $\STRUC$ motivates us to add a subset of \textit{initial vertices} $\mathring{V}\subseteq V$ to each simple graph $G=(V,E).$ If $f:G\rightharpoonup G'$ is a partial sieve-defined homomorphism of simple graphs, then we require
$
    \mathring{V} \subseteq \Def f,
    ~
    f(\mathring{V}) \subseteq \mathring{V}',
$ 
i.e., each partial sieve-defined homomorphism should be defined over initial vertices and preserve them. We generalize this situation and define an appropriate category.

\begin{definition}
    Let $\GraphSieve_{*}$ be a category whose objects are pairs $(G;\mathring{V}),$ where $G=(V,E)$ is a simple graph. Let $\mathring{V}\subseteq V$ be a set of \textit{initial vertices}, so that the morphisms of $\GraphSieve_{*}$ from $(G;\mathring{V})$ to $(G';\mathring{V}')$, denoted $((G;\mathring{V}),(G',\mathring{V}'))^{\partial}$, are (partial) maps $f\in\GraphSieve(G,G')$ such that 
    $
        \mathring{V} \subseteq \Def f,
        f(\mathring{V}) \subseteq \mathring{V}'. ~
    $ 
    The composition and identities are inherited from ones in $\GraphSieve.$ The set of everywhere-defined maps $f:G\to~G'$ such that $f(\mathring{V})\subseteq \mathring{V}'$ we denote by $((G;\mathring{V}),(G';\mathring{V}'))$.
\end{definition}

If a set of initial vertices is clear from the context, we can abbreviate $(G;\mathring{V})$ to just $G,$ still considering it as an object of  $\GraphSieve_{*}.$ Obviously, there is an inclusion
$
    ((G;\mathring{V}),(G',\mathring{V}'))
        \subseteq
    ((G;\mathring{V}),(G',\mathring{V}'))^{\partial}.
$

Let us consider the following covariant representable functor, which we use in an upcoming example (where the above inclusion becomes an equality),
\[
    (\N,\blank): 
    \begin{cases}
        \begin{tikzcd}[column sep= small, row sep=0ex]
            \GraphSieve_{*} 
                \arrow[r]
                &  
            \SetCategory
                \\
            G 
                \arrow[r, mapsto] 
                & 
            (\N,G)
                \\
            (f:G \rightharpoonup G')
                \arrow[r, mapsto] 
                & 
            (f\circ\blank:(\N,G)\to(\N,G'))
        \end{tikzcd}
    \end{cases}
\]
We will use the following example to represent and study paths in $\GraphSieve_{*}$ as executions (or traces) programs:

\begin{example}
    \label{ex:paths}
The natural ordering on the set of natural numbers $\N$ induces a simple graph $0 \to 1 \to \ldots,~$ which will be also denoted by $\N.$ By default, the set of initial vertices for $\N$ is set to be $\{0\},$ i.e., 
$
    (\N;\{0\}) \in \GraphSieve_{*}.
$
Moreover, given any $f:(\N;\{0\})\rightharpoonup (G;\mathring{V}),$ we get that 
$
    0 \in \Def f
$
and 
$    
    \Def f
$
is a sieve in 
$    
    G.
$
Since there exists a unique path in $\N$ from $0$ to any $n>0,$ any $n\in\N$ also belongs to $N.$ Therefore, $\Def f = \N,$ i.e., any partial sieve-defined homomorphism from $(\N;\{0\})$ to $(G;\mathring{V})$ is everywhere defined:
$
    ((\N;\{0\}),(G;\mathring{V}))
        =
    ((\N;\{0\}),(G;\mathring{V}))^{\partial}
$
or, skipping the sets of initial vertices, $(\N,G) = (\N,G)^{\partial}.$
Any $f:\N \to G$ models an infinite path $v_0 \to v_1 \to \ldots$ in the simple graph $G,$ where $v_i = f(i)$ and $v_0\in\mathring{V}.$ Of course, $(\N,G) = \varnothing,$ if $G$ is acyclic and finite. 
\end{example}

\begin{remark}
    When $G = (V,E)$ is a complete graph (including single loops at vertices), we identity $(\N,G)$ with a set 
     $
         \listsinf{V} :=
         \left\{
             \vec{v} = (v_0,v_1,\ldots)
             ~|~
             v_i \in V
         \right\} ~
     $
         of infinite lists of elements of $V.$ These arbitrary infinite lists of elements of $V$ are exactly the execution traces of program with states in $V$ \cite{liveness}.
     \end{remark}
\subsection{Graphs of $\LEDGER$, $\STRUC$, and their traces}

Recall that applying the functor $(\N,\blank)$ to a simple graph $G$ returns the set of infinite paths in $G$.
We use Example \ref{ex:paths} to define the abstract notion of a \emph{trace of a map}, which corresponds to the set of valid traces (considered as infinite paths in a state transition graph) that can be generated using this functor.

\begin{definition}
    Let $f:(G;\mathring{V})\rightharpoonup(G';\mathring{V}')$ be a morphism in $\GraphSieve_{*}.$ An image $\Im {f_{*}}$ of a postcomposition function $f_{*} = f \circ \blank: (\N \comma G) \to (\N \comma G')$ is called a \textit{trace} of $f.$ We'll denote the trace of $f$ by~$\Trc(f).$
\end{definition}

Unfolding the above definition, we get that the trace of $f:(G;\mathring{V})\rightharpoonup(G';\mathring{V}')$ consists of infinite paths 
$
    v'_0
        \to
    v'_1
        \to
    v'_2
        \to
    \ldots        ~
$
in $G',$ where $v'_0\in \mathring{V}',$ which has lifts in $G,$ i.e., an infinite path 
$
    v_0
        \to
    v_1
        \to
    v_2
        \to
    \ldots~
$
in $G,$ where $v_0\in\mathring{V},$ with
$
    f(v_0) = v'_0,
    ~~    
    f(v_1) = v'_1,
    ~~
    f(v_2) = v'_2,
    ~~
    \ldots ~
$
etc.

The definition of a valid $\LEDGER$ trace can be restated in terms of infinite paths in graphs. Let $\Lambda = (V,E)$ be a graph whose  set of vertices $V$ consists of  triples $(q,u,t)$ satisfying the condition $\fun{checkTx}$, and that there is an edge 
$
    (q,u,t) \to (q',u',t')~    
$
if and only if $(q,u,t,u')\in\LEDGER$ and $q\leq q'.$ The constructed graph $\Lambda$ is simple, as $u'$ is uniquely defined. Moreover, we specify the set of initial vertices to be 
$
    \mathring{V}:=
    \left\{
        (q,u,t) \in \Slot_0 \times \UTxO_0 \times \Tx 
            \mid
        \fun{checkTx}(q,u,t)
    \right\} ~
$

Similarly, we define a simple graph $\Lambda'=(V',E'),$ where $V'$ is a set of all $u\in\UTxO$ such that there are $q\in\Slot$ and $t\in\Tx$ with 
$
    \fun{checkTx}(q,u,t) = \True~
$
and there is an edge $u\to u'$ if and only if there are $q\in\Slot$ and $t\in\Tx$ with
$
    (q,u,t,u')\in\LEDGER. ~    
$
Simply speaking, $\Lambda'$ is a ``projection'' of $\Lambda$ on its $\UTxO\mhyphen$component. For the set of initial vertices in $\Lambda'$ we take
$
    \mathring{V}':= V' \cap \UTxO_0.~
$
The projection map $\varphi: V \to V'$ given by $\varphi(q,u,t) = u$
induces a morphism 
$
    \varphi: 
    (\Lambda; \mathring{V})
        \rightharpoonup
    (\Lambda'; \mathring{V}')~
$
in $\GraphSieve_{*}.$ Now, we apply the functor $(\N,\blank)$ 
to the morphism $\varphi:\Lambda\rightharpoonup\Lambda'$ we obtain a postcomposition function ${\varphi_{*}} = \varphi \circ \blank: (\N \comma \Lambda) \to (\N \comma \Lambda').$ Simply speaking, ${\varphi_{*}}$ takes an infinite path 
$
    (q_0,u_0,t_0)
        \to
    (q_1,u_1,t_1)
        \to
    (q_2,u_2,t_2)
        \to
    \ldots    ~
$
in $\Lambda,$ where $(q_0,u_0,t_0)\in\mathring{V},$ to an infinite path
$
    u_0
        \to
    u_1
        \to
    u_2
        \to
    \ldots~    
$
in $\Lambda'.$ 

Taking $\varphi = \LEDGER$, and recalling the definition of the set of valid $\LEDGER$ traces we obtain that 
$
    \Trc({\LEDGER}) = \Im {\varphi_{*}}
$

Let $(\STRUC,\pi,\kappa)$ be a structured contract implemented on $\LEDGER.$ Then, a similar construction exists for $\STRUC$. Let $\Gamma=(W,F)$ be a graph whose vertices are pairs $(s,i)\in\State\times\Input$ such that there exists $s'\in\State$ with $(*,s,i,s')\in\STRUC$ and there exists an edge $(s,i) \to (s',i')$ if and only if $(*,s,i,s')\in \STRUC.$ In the simple graph $\Gamma$ we specify the set of initial vertices 
\[
    \mathring{W}:=
    \left\{
        (s,i)\in \State_0 \times \Input 
        ~|~
        \exists~s'\in\State:~ (*,s,i,s')\in\STRUC
    \right\}  .~
\]

$\STRUC$ induces another graph $\Gamma'=(W',F')$ in the same way that $\Lambda$ is induced for $\LEDGER$. Here,  $W'$ is a set of all $s\in\State$ such that there are $i\in\Input$ and $s'\in\State$ with 
$
    (*,s,i,s')\in\STRUC ~
$
and there is an edge $s \to s'$ if and only if there are $i,i'\in\Input$ with 
$
    (s,i),(s',i')\in W'
$
and
$
    (*,s,i,s')\in\STRUC.
$
The set of initial vertices of $\Gamma$ is 
$
    \mathring{W}' := W'\cap \State_{0}. ~
$
Similarly to $\varphi:(\Lambda;\mathring{V})\rightharpoonup (\Lambda';\mathring{V}'),$ the projection map $\psi: W \to W'$ induces a morphism 
$
    \psi: 
    (\Gamma; \mathring{W})
        \rightharpoonup
    (\Gamma'; \mathring{W}')~
$
in $\GraphSieve_{*}.$ Again, applying the functor $(\N,\blankd):\GraphSieve_{*}\to\SetCategory$ to $\psi:\Gamma\rightharpoonup\Gamma',$ we obtain a postcomposition function ${\psi_{*}} = \psi \circ \blank: (\N \comma \Gamma) \to (\N \comma \Gamma'),$ 
which takes an infinite path 
$
    (s_0,i_0)
        \to
    (s_1,i_1)
        \to
    (s_2,i_2)
        \to
    \ldots ~
$
in $\Gamma,$ where $(s_0,i_0)\in\mathring{W},$ to an infinite path
$
    s_0
        \to
    s_1
        \to
    s_2
        \to
    \ldots    ~
$
in $\Gamma'.$ Taking $\psi = \STRUC$, and recalling the definition of the set of valid $\STRUC$ traces we obtain 
$
    \Trc({\STRUC}) = \Im {\psi_{*}}.~
$

Maps $\pi:\UTxO\rightharpoonup \State$ and $\kappa:\Tx\to\Input$ induce well-defined morphisms $\sigma(q,u,t) = (\pi u, \kappa t)$ and $\sigma'(u) = \pi u$ in $\GraphSieve_{*}$ such that the following diagram commutes:
\vspace{-3mm}
\[\begin{tikzcd}[ampersand replacement=\&,sep=2.25em]
	{(\Lambda;\mathring{V})} \& {(\Gamma;\mathring{W})} \\
	{(\Lambda';\mathring{V}')} \& {(\Gamma';\mathring{W}')}
	\arrow["\sigma", harpoon, from=1-1, to=1-2]
	\arrow["{\sigma'}", harpoon, from=2-1, to=2-2]
	\arrow["\varphi"', harpoon, from=1-1, to=2-1]
	\arrow["\psi", harpoon, from=1-2, to=2-2]
\end{tikzcd}
\vspace{-3mm}\]

This shows that when $\pi$ and $\kappa$ specify a correct implementation of the $\STRUC$ program on the ledger, a (valid) state trace  $\vec{v} \in \Trc(\LEDGER)$, corresponding to a specific lift (i.e., a path in graph that contains both state and input information), maps to a (valid) trace of $\vec{w} \in \Trc(\STRUC)$. Moreover, $\vec{w}$ necessarily corresponds to the lift that is obtained by applying the function induced by $\pi$ and $\kappa$ to the lift of $\vec{v}$. In other words, $\vec{v}$ and $\vec{w}$ are "generated" by corresponding (via $\kappa$) sequences of inputs.
This conclusion is the key consequence of fulfilling the proof obligation \ref{def:structured} required as part of instantiating $(\STRUC, \pi, \kappa)$. 

Our goal is now to prove that the above square induces a well-defined continuous postcomposition function $\overline{\pi}:\Trc({\LEDGER}) \to \Trc({\STRUC})$, thereby ensuring that correct program implementations are also well-behaved with respect to certain kinds of properties. To achieve this goal we will study a specific kind of metric.
\subsection{Ultrametric spaces and traces}

Sets of morphisms $(\N,G)$ in $\GraphSieve_{*}$ also carry an additional structure: they are metric spaces. A metric $d$ on the set $(\N,G)$ is 
$
    d(\vec{a},\vec{b}) 
        :=    
    \inf 
    \left\{ 
        2^{-k} 
        ~|~ 
        k\geq 0, a_i \neq b_i  
    \right\}.~
$
The function $d$ also satisfies a strengthened version of the triangle inequality:
$
    d(\vec{a},\vec{b}) 
        \leq    
    \max\{d(\vec{a},\vec{c}), d(\vec{c},\vec{b})\}.~
$

This metric has been previously defined for arbitrary execution traces of arbitrary programs\cite{liveness}. 
In this work, we apply this definition to the space of only valid execution traces. Spaces admitting the type of metric given above are said to be \textit{ultrametric} \cite{ultrametric}. Ultrametric spaces have the following properties, which we will later make use of:

\begin{proposition}
    Let $(X,d)$ be an ultrametric space. Then, (i) any triangle in $X$ is isosceles, i.e., for any $x,y,z\in X,$ two of the numbers $d(x,y),d(y,z),d(x,z)$ are equal and greater than the third one; (ii) every point inside an open ball is its center; (iii) if two open balls intersect, then one is contained in another; (iv) every open ball of a positive radius is a closed set.
\end{proposition}

From the above proposition following properties of the ultrametric space can be derived $((\N,G),d)$: 
(i) $d$ takes values in a countable set $\{0,1,\frac{1}{2},\frac{1}{4},\ldots\};$ 
(ii) the whole space coincides with a closed unit ball, i.e., $(\N,G) = \overline{B}(\vec{v},1)$ for any $\vec{v}\in(\N,G);$ 
(iii) for any $\vec{u}\in(\N,G)$ and $r > 0,$
    $
      B(\vec{u}, r) = 
      B(\vec{u},2^{-\underline{r}}) = 
      B(\vec{u},2^{-(\underline{r}+1)}),~
    $
where $\underline{r} = \left\lfloor -\log_2 r\right\rfloor$, i.e., there are only a countable number of balls with a center~$\vec{u};$ (iv) for any $\vec{u}\in(\N,G)$ and $r > 0,$
$
    B(\vec{u}, r) := 
    \left\{
    \vec{v}\in(\N,G)
    ~|~
    \vec{v}[n] = \vec{u}[n]
    \right\}, ~
$
where $\vec{u}[n]$ is a head of the path of $\vec{u}$ of length $n$ in $G$,  
i.e., a ball of radius $r$ with a center $\vec{u}$ consists of infinite paths in $G,$ that share an $\underline{r}\mhyphen$ head with $\vec{u}.$ From now we will treat $(\N,G)$ as an ultrametric space with the described above ultrametric $d.$

The following proposition states that morphisms in $\GraphSieve_{*}$ induce functions between sets of infinite paths that are, in fact, continuous maps of ultrametric spaces. To prove this, we will use the fact that the induced map $f_{*}$ simply applies $f$ to each individual state in the infinite path.

\begin{proposition}
If $f:G\rightharpoonup G'$ is a morphism in $\GraphSieve_{*},$ then 
$
    f_{*}:(\N,G) \to (\N,G')~
$
is a continuous map. 
\end{proposition}

\begin{proof}
    For arbitrary $\vec{u}$ in $(\N,G),$ and $n\geq 0,$ take any $\vec{v} \in B(f_{*}\vec{u}, 2^{-n}).$ Then $\vec{v}$ has form
    $
        fu_0
            \to
        \ldots
            \to
        fu_n
            \to
        v_{n+1}
            \to
        \ldots  .~  
    $
    For any $\vec{w}\in B(\vec{u}, 2^{-n}),$ its image under $f_{*}$ is
    $
        fu_0
            \to
        \ldots
            \to
        fu_n
            \to
        fw_{n+1}
            \to
        \ldots    .~ 
    $
    Hence, $f\vec{w}\in B(f_{*}\vec{u},2^{-n}).$ and $f_{*}$ is continuous.     
\end{proof}

In order to study ultrametric structure in $\GraphSieve_{*}$, we define the following category:

\begin{definition}
    A \textit{non-expanding map} between ultrametric spaces $(X,d_X)$ and $(Y,d_Y)$ is a continuous map $f:X\to Y$ such that
$
    d_Y(f(x_1),f(x_2))
        \leq
    d_X(x_1,x_2)~
$
for any $x_1,x_2\in X.$ A category of ultrametric spaces and non-expanding maps we denote~$\UMet.$
\end{definition}

We now use the $\UMet$ category to express that all maps between sets of infinite paths in objects of our category $\GraphSieve_{*}$ (i.e., simple graphs) are continuous and non-expanding:

\begin{proposition}
    \label{prop:continuous}
    A functor $(\N,\blankd):\GraphSieve_{*}\to\SetCategory$ factors through the category $\UMet.$
\end{proposition}

\begin{proof}
    The only thing to check is that $f_{*}:(\N,G)\to(\N,G')$ is non-expanding, for any $f:G\rightharpoonup G'.$ Take any $\vec{u},\vec{v}\in(\N,G).$ Suppose that 
$
    d(\vec{u},\vec{v}) = 2^{-n}.~
$
Hence, there are vertices $w_0,\ldots,w_n$ of $G$ such that 
$
    \vec{u} = 
    \left(
        w_0 \to
        \ldots \to
        w_n \to
        u_{n+1} \to
        \ldots
    \right)    
$
~ and ~
$
    \vec{v} = 
    \left(
        w_0 \to
        \ldots \to
        w_n \to
        v_{n+1} \to
        \ldots
    \right)    .~
$
Applying $f_{*}$ to $\vec{u}$ and $\vec{v},$ we obtain that $f_{*}\vec{u}$ and $f_{*}\vec{v}$ share a common head of length at least $n,$ i.e., 
$
    d(f_{*}\vec{u},f_{*}\vec{v}) \leq 2^{-n} = d(\vec{u},\vec{v})~
$
and $f_{*}$ is non-expanding.
\end{proof}

The above proposition allows us to consider $(\N,\blankd):\GraphSieve_{*}\to\SetCategory$ as a functor $(\N,\blankd):\GraphSieve_{*}\to\UMet$ instead.
We introduce some additional concepts for reasoning about program behaviour,
including executions, properties, and safety, which we can interpret within the framework we defined here.
Infinite lists of elements of set $S$ are called \textit{executions} \cite{liveness}. They represent sequences of states of some process, for whom $S$ serve as a set of all possible states. The subset $P\subseteq\listsinf{S}$ is identified with its characteristic function $\chi_P:\listsinf{S}\to\Bool,$ and the latter is called a \textit{property}. 

A property $P$ is said to be a \textit{safety property}, if $P$ does not hold for an execution $\vec{s},$ then at some state $s_i$ some ``\textit{bad thing}'' must happen. Such a ``bad thing'' must be irremediable because a safety property states that the ``bad thing'' never happens during execution. In other words, we characterize $P\subseteq \listsinf{S}$ via a statement about its complement:
\[
    \vec{s}\notin P 
    ~~\Leftrightarrow~~
    (\exists~n\geq 0)
    (\forall \vec{t}\in\listsinf{S})
    \{(s_0,\ldots,s_n,t_0,t_1,\ldots)\notin P\}
\]
Using the ultrametric $d$ for $\listsinf{S},$ $\vec{s}\notin P$ is equivalent to 
$
    (\exists~ n\geq 0)
    \{
        B(\vec{s},2^{n}) \cap P = \varnothing
    \}.~
$
In other words, the complement of $P$ contains every point with some its neighborhood. So the safety property $P$ uniquely defines a closed subset in $(\listsinf{S},d)$ and vice versa. 

Another important class of properties of execution traces is \emph{liveness} \cite{liveness}, which, like safety, can also be defined in topological terms. An arbitrary property of traces can be defined in terms of a liveness and safety property. The structured contract formalism, as well as blockchains in general, however, are more amenable to being studied from the point of view of guaranteeing safety. Investigating liveness in this setting is equally important, but poses more of a challenge. We, therefore, leave it for future work. 


Recall also that we defined the trace $\Trc(f)$ of $f:(G;\mathring{V})\rightharpoonup(G';\mathring{V}')$ as the image of $f_{*}$, i.e., a subset of $(\N \comma G')$. Recall also that the representable functor $(\N,\blank)$ factors through $\UMet$, inducing an ultrametric in $(\N \comma G')$, so that
traces of morphisms in $\GraphSieve_{*}$ are themselves ultrametric spaces. That is, we endow the subset
$
    \Im(f_{*}) \subseteq (\N,G')~
$
with the ultrametric structure induced via metric from $(\N,G')$.

We obtain that a map in $\GraphSieve_{*}$ bewteen two graphs induces a morphism between sets of infinite paths in the those graphs in a way that preserves the ultrametric structure as it is non-expanding:

\begin{lemma}
    A commutative square
    \vspace{-3.5mm}
    \[\begin{tikzcd}[ampersand replacement=\&,sep=2.25em]
        {(G;\mathring{V})} \& {(H;\mathring{W})} \\
        {(G';\mathring{V}')} \& {(H';\mathring{W}')}
        \arrow["p", harpoon, from=1-1, to=1-2]
        \arrow["{q}", harpoon, from=2-1, to=2-2]
        \arrow["f"', harpoon, from=1-1, to=2-1]
        \arrow["g", harpoon, from=1-2, to=2-2]
    \end{tikzcd}
    \vspace{-3.5mm}
    \]
    in $\GraphSieve_{*}$ induces a non-expanding postcomposition function $\overline{q}:  \Trc(f) \to  \Trc(g).$ 
\end{lemma}

\begin{proof}
    Applying the functor $(\N,\blankd):\GraphSieve_{*}\to\UMet$ to a given square, we obtain a commutative diagram of non-expanding maps
    \vspace{-3.5mm}
\[\begin{tikzcd}
	{(\N,G)} &&& {(\N,H)} \\
	& {\Trc(f)} &&& {\Trc(g)} \\
	{(\N,G')} &&& {(\N,H')}
	\arrow["{p_{*}}", from=1-1, to=1-4]
	\arrow[two heads, from=1-1, to=2-2]
	\arrow["{f_{*}}"', from=1-1, to=3-1]
	\arrow[two heads, from=1-4, to=2-5]
	\arrow["{g_{*}}"{pos=0.3}, from=1-4, to=3-4]
	\arrow["{\overline{q}}", dashed, from=2-2, to=2-5]
	\arrow[hook', from=2-2, to=3-1]
	\arrow[hook', from=2-5, to=3-4]
	\arrow["{q_{*}}", from=3-1, to=3-4]
\end{tikzcd}
\vspace{-3.5mm}
\]
\newline For any $\vec{s}\in\Trc{f}=\Im f_{*}$ there is $\vec{v}\in(\N,G)$ such that
$
    f_{*}(\vec{v}) = \vec{s}    .~
$
Since $q_{*}f_{*} = g_{*}p_{*},$ then
$
    \overline{q}(\vec{s}) = 
    q_{*}(\vec{s}) = 
    q_{*}(f_{*}(\vec{v})) = 
    g_{*}(p_{*}(\vec{v})) \in \Im g_{*}.~
$
Hence, we checked that $\overline{q}$ is well-defined, i.e., its codomain is $\Trc(g).$ 
Since the inclusion of a subspace $\Im (f_{*})$ in $(\N, G')$ is a non-expanding map, so is $\overline{q}$ as composition of the inclusion and $q_{*}.$
\end{proof}

Applying the above result to structured contracts, we get:

\begin{corollary}[\textbf{Trace-mapping lemma}]
    A structured contract $(\STRUC, \pi, \kappa)$ for $\LEDGER$ induces a non-expanding map 
\[
    \overline{\pi}:
    \Trc({\LEDGER})
        \to
    \Trc({\STRUC}) ~
\]
between ultrametric spaces.
\end{corollary}

The trace-mapping lemma expresses a practical result of this work: a safety property 
(i.e., a closed subset of valid traces) of $\Trc({\STRUC})$ necessarily has a corresponding safety property in $\Trc({\LEDGER})$, 
i.e., a closed subset of in the domain of the non-expanding map induced by $(\STRUC, \pi, \kappa)$. 
In particular, $\Trc({\STRUC})$, which is also a closed subset of itself, has a closed preimage in $\Trc({\LEDGER})$.
This preimage is a safety property of $\Trc({\LEDGER})$.
This safety property can be expressed in the "a specific bad thing cannot be fixed if it happens" language as: 
if a trace $\vec{v}\in \Trc(\LEDGER)$ is not in the preimage of $\Trc(\STRUC)$, there is necessarily some $i$, such 
that $v_0~\to~\ldots~\to~v_i$ is a prefix of $\vec{v}$, and no trace of the form $v_0~\to~\ldots~\to~v_i~\to u_{i+1}~\to~\ldots $ maps to a trace in $\Trc({\STRUC})$.

Preimages of traces also carry important information about them. For this reason, we define the concept of a \textit{full lift} of a trace:

\begin{definition}
Each element $\vec{s}$ of $\Trc(f) = \Im(f_{*})$ has \textit{full lift} $\vec{v}\in(\N,G)$ via 
$
    f_{*}(\vec{v}) = \vec{s} .~
$
\end{definition}

In other words, for every trace $\vec{s}$ of $f,$ which is an infinite path, there is an infinite path $\vec{v}$ that is mapped to it. The latter yields that each head $\vec{s}[n]$ of $\vec{s}$ has a lift $\vec{v}[n]$ via
$
    f_{*}(\vec{v}[n]) = \vec{s}[n] .~
$
In the following proposition, we express when the converse holds,

\begin{proposition}
    Let $f:G\rightharpoonup G'$ be a morphism in $\GraphSieve_{*}.$ Then 
    $
        \overline{\Trc(f)} 
        =
        \bigcap_{n\geq 0} \Trc_{n}(f)     ,~
    $
    where $\overline{\Trc(f)}$ is the closure of $\Trc(f)$ in $(\N,G')$ and 
    $
        \Trc_n(f) :=
        \{
            \vec{s}:\N\to G'
            ~|~
            \exists~ \vec{v}:\N\to G:~ f_{*}(\vec{v}[n]) = \vec{s}[n]
        \}    ~
    $
    is a closed set of infinite paths in $G'$ that has \textit{$n\mhyphen$truncated lifts} in $(\N,G).$ Hence, $\Trc(f)$ is closed if and only if the following holds:
    any infinite path in $G'$ has a full lift in $G$ if and only if an infinite path in $G'$ has an $n\mhyphen$truncated lift in $G$ for all $n\geq 0$.   
\end{proposition}

\begin{proof}
    Let $\vec{s}$ be a limit point of $\Trc_n(f).$ Then 
$
    \exists~ \vec{t}\in B(\vec{s},2^{-n}) \cap \Trc_n(f) \neq \varnothing.~
$
Hence, 
$
    \vec{t} = (
        s_0 \to
        \ldots \to
        s_n \to
        t_{n+1} \to
        t_{n+2} \to
        \ldots
    ) ~
$
is an infinite path in $G'$ and there is $\vec{v}\in (\N,G)$ such that 
$
    f_{*}(\vec{v}[n]) = \vec{t}[n]    ~
$
i.e., $t_i = s_i = f(v_i)$ for $0\leq i \leq n.$ Therefore,
$
    f_{*}(\vec{v}[n]) = \vec{s}[n]~
$
so $\vec{s}\in\Trc_n(f)$ and $\Trc_n(f)$ is closed.
Obviously, $\Trc(f) \subseteq \Trc_n(f)$ for every $n\geq 0,$ so $\Trc(f)$ is a subset of $\bigcap_{n\geq 0}\Trc_n(f),$ which is a closed set. Hence,
$
    \overline{\Trc(f)} 
        \subseteq
    \bigcap_{n\geq 0}\Trc_n(f) ~
$
as the closure is the smallest closed set containing $\Trc(f).$ 

Conversely, if $\vec{s}\in \Trc_n(f)$ for all $n\geq 0,$ then 
$
    (\forall~ n\geq 0)
    (\exists~ \vec{v}\in(\N,G))    
    \{
        f_{*}(\vec{v}[n]) = \vec{s}[n]
    \} .~
$
The latter means that $d(f_{*}(\vec{v}), \vec{s})<2^{-n}$ for each $n\geq 0.$ But $f_{*}(\vec{v})\in\Im(f_{*}),$ so
$
    B(\vec{s},2^{-n}) \cap \Trc(f) \neq \varnothing    ~
$
for each $n\geq 0.$ Therefore, $\vec{s}$ is a limit point of $\Trc(f)$ and 
$
    \bigcap_{n\geq 0}\Trc_n(f)
        \subseteq
    \overline{\Trc(f)} ~
$
which was desired.
\end{proof}
\section{Properties of LEDGER system}


Since we do not give implementation details or concrete examples of structured contracts, presenting examples of their properties is left for future work. In this section, we focus on properties of the $\LEDGER$ system itself, for which we first introduce additional notation and terminology.

The set $\UTxO$ is assumed to be ``\textit{well-founded}'': for any $u\in\UTxO_0$ and any key-value pair $((b,n),o)\in u,$ there is a transaction $t\in T$ such that 
$
    b = h(t),~ \inputs(t) = \varnothing    
$  
i.e. hash part of any key in a valid initial UTxO state $u\in\UTxO_0$ comes from a transaction with empty list of inputs. In the rest of the paper we'll assume that $\UTxO$ is well-founded. 

An updated state $u'~=~(u \setminus \fun{getORefs}(t)) \sqcup \fun{mkOuts}(t)$ is constructed from a triple $(q,u,t)$ satisfying $\fun{checkTx}$ via the functions $\fun{getORefs}$ and $\fun{mkOuts}.$ The condition $\fun{checkTx}$ guarantees that all triples $(b,n,o)$ extracted from $\inputs(t)$ belong to~$u.$ After subtracting $\fun{getORefs}(t)$ from $u,$ the function $\fun{mkOuts}$ joins new pairs 
$
    \{
        (h(t),n,o)
        ~|~
        (n,o) \in \outputs(t)
    \}    
$
to the result from the previous step. For a lift
$
    (q_0,u_0,t_0) \to
    (q_1,u_1,t_1) \to
    \ldots 
$
of a valid $\LEDGER$ trace 
$
    u_0 \to u_1 \to \ldots u_k \to \ldots 
$
we introduce the following notations:
$
    r_i := \fun{getORefs}(t_i), ~
    c_i := \fun{mkOuts}(t_i). 
$
Therefore, the sequence of valid $\LEDGER$ states $u_{k+1} = (u_k \setminus r_k) \cup c_k,$ for $k\geq 0,$ where $u_0\in\UTxO_0.$
\subsection{Replay and trivial update protection}

An important property of UTxO ledgers is that an attacker is not able to disrupt the operation of the ledger program by re-applying an existing transaction. A related property is that it is not possible to apply a transaction that does not change the ledger state. We can formalize these as safety properties of ledger traces in the following way:

\begin{theorem}[(Replay and trivial update protection]
    Given an infinite path $(q_0,u_0,t_0) \to (q_1,u_1,t_1) \to \ldots$ in a graph corresponding to $\LEDGER,$ for any indexes $i<j$
    \begin{itemize}
        \item[(a)] $t_i\neq t_j$ (replay protection);
        \item[(b)] $u_i \neq u_j$ (trivial update protection).
    \end{itemize}
    Both (a) and (b) are safety properties of $\Trc({\LEDGER})$.
\end{theorem}
\begin{proof}
    (a) Suppose $i<j$ is the minimal pair such that $t_i = t_j.$ Then 
$
    r_i \subseteq u_i = (u_{i-1} \setminus r_{i-1}) \cup c_{i-1} \subseteq u_{i-1} \cup c_{i-1} \subseteq u_0 \cup c_0 \cup c_1 \cup \ldots \cup c_{i-1}.
$
A similar chain of inclusions yields
$
    r_j ~\subseteq~ 
    u_j ~\subseteq~
    (u_i \setminus r_i) \cup 
    c_i \cup \ldots \cup c_{j-1}.~
$
Since $r_i = r_j,$ the set $r_j$ has no common elements with $u_i \setminus r_i,$ so 
$
    r_j ~\subseteq~ 
    c_i \cup \ldots \cup c_{j-1}
$
and 
$
    r_i \cap r_j ~\subseteq~
    \left(
        u_0 \cup c_0 \cup \ldots \cup c_{i-1}
    \right)
    \cap 
    \left(
        c_i \cup \ldots \cup c_{j-1}
    \right).
$

\textbf{Claim 1.} For any $l\geq 0:$ $u_0 \cap c_l = \varnothing.$ Suppose $(b,n,o)\in u_0 \cap c_l,$ for some $l\geq 0.$ Then $b$ is a hash of a transaction $\mathring{t}$ with an empty input, by the well-foundedness of $\UTxO.$ On the other hand, $b$ is a hash of a transaction $t_l$ with a non-empty input, since $(q_l,u_l,t_l)$ satisfies $\fun{checkTx}.$ Hence, $\mathring{t} = t_l,$ by the injectivity of $h,$ which is a contradiction, as their inputs differ. So, $u_0 \cap c_l = \varnothing.$

\textbf{Claim 2.} For any $0\leq l \leq i-1$ and $i \leq m \leq j-1:$  $c_l \cap c_m = \varnothing.$ Hash components of all elements in $c_l$ and $c_m$ are $h(t_l)$ and $h(t_m),$ respectively. If these sets share an element, hash of this element is $h(t_l) = h(t_m).$ From the injectivity of $h$ follows that $t_l = t_m.$ The latter equality contradicts the assumption about the minimality of the pair $(i,j).$ Therefore, $c_l \cap c_m = \varnothing.$

Claims 1 and 2 imply $r_i \cap r_j = \varnothing,$ which means that $r_i = r_j = \varnothing.$ Hence, $\inputs(t_i) = \inputs(t_j) = \varnothing,$ that contradicts the fact that $(q_i,u_i,t_i)$ and $(q_j,u_j,t_j)$ satisfy $\fun{checkTx}.$

(b) Suppose there is a pair $i < j$ such that $u_i = u_j.$ An equality of sets 
$
    u_j
    =
    (u_{j-1}\setminus r_{j-1}) \cup c_{j-1}    
$
imply that 
$
    c_{j-1} ~\subseteq~
    u_{i} ~\subseteq~
    u_{i-1} \cup c_{i-1} ~\subseteq~
    u_{i-2} \cup c_{i-2} \cup c_{i-1} ~\subseteq~
    \ldots
    u_0 \cup c_0 \cup \ldots \cup c_{i-1}
$
By Claim 1 $u_{0}\cap c_{j-1} = \varnothing.$ If $c_{j}\cap c_{l} \neq \varnothing,$ where $0 \leq l \leq i-1,$ then $t_j = t_l,$ which is impossible by part (a). Hence, 
$
    (u_0 \cup c_0 \cup \ldots \cup c_{i-1}) \cap c_{j-1} = \varnothing
$  
and $u_i \neq u_j.$

\textbf{Claim 3.} Because the full space (which is always closed) satisfies this property, this property represents a closed set, and is therefore a safety property.
\end{proof}

An important corollary of the above theorem is that sets we add to or delete from a given ledger state (UTxO set) are pairwise disjoint. This will be used in an upcoming result.

\begin{corollary}
    \label{cor:injective}
If $h:\Tx\to\ByteString$ is injective, then: (i) sets $u_0,~ c_0,~ c_1,~ \ldots$ are pairwise disjoint; (ii) sets $r_0,r_1,\ldots$ are pairwise disjoint. Moreover, 
$
    u_{k+1} = (u_k \setminus r_k) \sqcup c_k    
$
for any $k\leq 0.$
\end{corollary}

We note here that the above properties would also be safety properties when considered as properties of arbitrary (not necessarily valid) traces. This can be justified as follows: any trace containing a prefix such that for $i\neq j$, transactions $t_i = t_j$ will never satisfy the replay protection property, regardless of its suffix. The "bad thing" cannot be fixed. Similarly, a trace containing a trivial update at states $u_i = u_j$ cannot be "fixed" by any suffix.  
\subsection{UTxO transaction commutativity}

The $\UTxO$ set in the $\LEDGER$ transition system enjoys the transaction commutativity property meaning that \textit{the order of applying transactions to a valid initial state is irrelevant and always returns the same result.} Still we should keep in mind that every single transaction application must be validated by $\fun{checkTx}.$ While this property has to do with finite sequences of states and transactions, we can express it a property of (infinite) execution traces.

\begin{theorem}[UTxO transaction commutativity]
    Let $u_0$ is a well-founded UTxO state. Suppose we have two traces with (possibly) distinct length-$n+1$ prefixes, and an arbitrary suffix,

\vspace{-4mm}    
\[\begin{tikzcd}[ampersand replacement=\&,column sep=2.25em,row sep=tiny]
	{(q_0,u_0,t_0)} \& {(q_1,u_1,t_1)} \& \ldots \& {(q_n,u_n,t_n)} \& {s_1} \& {s_2} \& \ldots  \\
	{(q'_0,u'_0,t'_0)} \& {(q'_1,u'_1,t'_1)} \& \ldots \& {(q'_n,u'_n,t'_n)} \& {s_1'} \& {s_2'} \& \ldots
	\arrow[from=1-1, to=1-2]
	\arrow[from=1-2, to=1-3]
	\arrow[from=1-3, to=1-4]
	\arrow[from=1-4, to=1-5]
    \arrow[from=1-5, to=1-6]
    \arrow[from=1-6, to=1-7]
	\arrow[from=2-1, to=2-2]
	\arrow[from=2-2, to=2-3]
	\arrow[from=2-3, to=2-4]
    \arrow[from=2-4, to=2-5]
    \arrow[from=2-5, to=2-6]
    \arrow[from=2-6, to=2-7]
\end{tikzcd}
\vspace{-3mm}
\]
in the simple graph corresponding to $\LEDGER,$ where $u_0=u'_0\in\UTxO_0$ and $(t'_0,\ldots,t'_n)$ is a permutation of $(t_0,\ldots,t_n).$ If $h$ is injective, then $u_n = u'_n.$ 
\end{theorem}

Before we prove the above result, we motivate the proof and illustrate this property with an example:

\begin{example}
    Let $(t_0,t_1,t_2,t_3,t_4,t_5,t_6,t_7)$ and $(t_3,t_1,t_6,t_2,t_5,t_7,t_0,t_4)$ be valid sequences of transactions that are applied to a valid initial state $u_0\in\UTxO_0.$ Then 
\begin{center}
\begin{tabular}{rclcl|rclcl}
    $u_1$ & $=$ & $(u_0 \setminus r_0) \sqcup c_0$ & $\subseteq$ & $u_0\sqcup c_0$ 
    &
    $u'_1$ & $=$ & $(u_0 \setminus r_3) \sqcup c_3$ & $\subseteq$ & $u_0\sqcup c_3$ 
    \\
    $u_2$ & $=$ & $(u_1 \setminus r_1) \sqcup c_1$ & $\subseteq$ & $u_0\sqcup c_{01}$
    &
    $u_2$ & $=$ & $(u'_1 \setminus r_1) \sqcup c_1$ & $\subseteq$ & $u_0\sqcup c_{13}$ 
    \\
    $u_3$ & $=$ & $(u_2 \setminus r_2) \sqcup c_2$ & $\subseteq$ & $u_0\sqcup c_{012}$ 
    &
    $u'_3$ & $=$ & $(u'_2 \setminus r_6) \sqcup c_6$ & $\subseteq$ & $u_0\sqcup c_{136}$
    \\
    $u_4$ & $=$ & $(u_3 \setminus r_3) \sqcup c_3$ & $\subseteq$ & $u_0\sqcup c_{0123}$
    &
    $u'_4$ & $=$ & $(u'_3 \setminus r_2) \sqcup c_2$ & $\subseteq$ & $u_0\sqcup c_{1236}$ 
    \\
    $u_5$ & $=$ & $(u_4 \setminus r_4) \sqcup c_4$ & $\subseteq$ & $u_0\sqcup c_{01234}$ 
    &
    $u'_5$ & $=$ & $(u'_4 \setminus r_5) \sqcup c_5$ & $\subseteq$ & $u_0\sqcup c_{12356}$ 
    \\
    $u_6$ & $=$ & $(u_5 \setminus r_5) \sqcup c_5$ & $\subseteq$ & $u_0\sqcup c_{012345}$
    &
    $u'_6$ & $=$ & $(u'_5 \setminus r_7) \sqcup c_7$ & $\subseteq$ & $u_0\sqcup c_{123567}$
    \\
    $u_7$ & $=$ & $(u_6 \setminus r_6) \sqcup c_6$ & $\subseteq$ & $u_0\sqcup c_{0123456}$
    &
    $u'_7$ & $=$ & $(u'_6 \setminus r_0) \sqcup c_0$ & $\subseteq$ & $u_0\sqcup c_{0123567}$
    \\
    $u_8$ & $=$ & $(u_7 \setminus r_7) \sqcup c_7$ & $\subseteq$ & $u_0\sqcup c_{01234567}$
    &
    $u'_8$ & $=$ & $(u'_7 \setminus r_4) \sqcup c_4$ & $\subseteq$ & $u_0\sqcup c_{01234567}$
    \\
\end{tabular}
\end{center}
where $c_{i_1i_2\ldots i_k} = \cup_{j=1}^{k}c_{i_j}.$ Since $c_A \cap c_B = c_{A\cap B}$ for any subsets $A,B$ of $\{0,\ldots,7\},$ 
\begin{center}
    \begin{tabular}{rclcl}
        $r_0$ & $\subseteq$ & $u_0 \cap (u_0 \sqcup c_{123567})$ & $=$ & $u_0$ \\ 
        $r_1$ & $\subseteq$ & $(u_0 \sqcup c_0) \cap (u_0 \sqcup c_{3})$ & $=$ & $u_0$ \\ 
        $r_2$ & $\subseteq$ & $(u_0 \sqcup c_{01}) \cap (u_0 \sqcup c_{136})$ & $=$ & $u_0 \sqcup c_1$ \\ 
        $r_3$ & $\subseteq$ & $(u_0 \sqcup c_{012}) \cap u_0$ & $=$ & $u_0$ \\ 
        $r_4$ & $\subseteq$ & $(u_0 \sqcup c_{0123}) \cap (u_0 \sqcup c_{0123567})$ & $=$ & $u_0 \sqcup c_{0123}$ \\ 
        $r_5$ & $\subseteq$ & $(u_0 \sqcup c_{01234}) \cap (u_0 \sqcup c_{1236})$ & $=$ & $u_0 \sqcup c_{123}$ \\ 
        $r_6$ & $\subseteq$ & $(u_0 \sqcup c_{012345}) \cap (u_0 \sqcup c_{13})$ & $=$ & $u_0 \sqcup c_{13}$ \\
        $r_6$ & $\subseteq$ & $(u_0 \sqcup c_{012345}) \cap (u_0 \sqcup c_{13})$ & $=$ & $u_0 \sqcup c_{13}$ \\ 
        $r_7$ & $\subseteq$ & $(u_0 \sqcup c_{0123456}) \cap (u_0 \sqcup c_{12356})$ & $=$ & $u_0 \sqcup c_{12356}$ \\
    \end{tabular}
\end{center}
The above inclusions yield: 

\vspace{-2mm}

\begin{itemize}
    \item[(0)] subsets $r_0,r_1,r_3$ are subtracted from $u_0$ in any order, then $c_0,c_1,c_3$ are attached; \vspace{-2mm}
    \item[(1)] subsets $r_2,r_6$ are subtracted from the result of (0) in any order, then $c_2,c_6$ are attached; \vspace{-2mm}
    \item[(2)] subsets $r_4,r_5$ are subtracted from the result of (1) in any order, then $c_4,c_5$ are attached; \vspace{-2mm}
    \item[(3)] subset $r_7$ is subtracted from the result of (2), then $c_7$ is attached. \vspace{-2mm}
\end{itemize}
Hence, the set $\{0,\ldots,7\}$ is equipped with a partial order structure: $i<j$ if the subtraction of $r_i$ can potentially depend on attaching of $c_j.$ In the case of our example, the Hasse diagram of a poset is
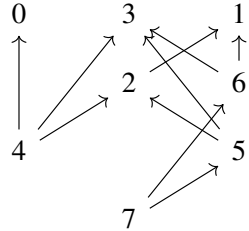
\[\begin{tikzcd}[ampersand replacement=\&,row sep=small]
	0 \& 3 \& 1 \\
	\& 2 \& 6 \\
	4 \&\& 5 \\
	\& 7
	\arrow[from=3-1, to=1-1]
	\arrow[from=3-1, to=2-2]
	\arrow[from=2-2, to=1-3]
	\arrow[from=3-3, to=1-2]
	\arrow[from=2-3, to=1-2]
	\arrow[from=2-3, to=1-3]
	\arrow[from=3-3, to=2-2]
	\arrow[from=3-1, to=1-2]
	\arrow[from=4-2, to=3-3]
	\arrow[from=4-2, to=2-3]
\end{tikzcd}\]
Vertices at the top of the diagram (sinks) are \textit{level 0} vertices. A level of a vertex $i$ is a length of a maximal path from $i$ to vertices of level 0. Hence, we obtain a level partition of the set of vertices: (0) level 0 vertices: 0,1,3; (1) level 1 vertices: 2,6; (2) level 2 vertices: 4,5; (3) level 3 vertex: 7.   

From our construction follows that transactions labeled by vertices of the same level commute, and a pair $(t_i,t_j)$ can be swapped to $(t_j,t_i),$ if the level of $i$ is less than the level of $j.$ 

Fixing the ordering between the initial set of transactions, we obtain a canonical presentation of the set of transactions
$
    (0,1,3 ~|~ 2,6 ~|~ 4,5 ~|~ 7)    
$  
and the canonical path, where $w_0 = u_0$:
\[
    (w_0, t_0) \to
    (w_1, t_1) \to
    (w_2, t_3) \to
    (w_3, t_2) \to
    (w_4, t_6) \to
    (w_5, t_4) \to
    (w_6, t_5) \to
    (w_7, t_7) 
\]

Finally, we present a way how to transform one valid transaction sequence into another. Indeces of numbers represent their levels, adjacent red numbers are swapped in the next row.
\begin{center}
\begin{tabular}{ccccccccccccccccc}
    $(0_0$ & $1_0$ & \red{$2_1$} & \red{$3_0$} & $4_2$ & \red{$5_2$} & \red{$6_1$} & $7_3)$ &
    $~~$ & 
    $($\red{$3_0$} & \red{$1_0$} & $6_1$ & $2_1$ & 
    $5_2$ & \red{$7_3$} & \red{$0_0$} & $4_2)$ 
    \\
    $(0_0$ & $1_0$ & $3_0$ & $2_1$ & \red{$4_2$} & \red{$6_1$} & {$5_2$} & $7_3)$ &
    $~~$ & 
    $(${$1_0$} & {$3_0$} & \red{$6_1$} & \red{$2_1$} & 
    $5_2$ & {$0_0$} & \red{$7_3$} & \red{$4_2$}$)$ 
    \\
    $(0_0$ & $1_0$ & $3_0$ & $2_1$ & {$6_1$} & {$4_2$} & {$5_2$} & $7_3)$ &
    $~~$ & 
    $(${$1_0$} & {$3_0$} & {$2_1$} & {$6_1$} & 
    \red{$5_2$} & \red{$0_0$} & {$4_2$} & {$7_3$}$)$ 
    \\
    $~$ & $~$ & $~$ & $~$ & {$~$} & {$~$} & {$~$} & $~$ &
    $~~$ & 
    $(${$1_0$} & {$3_0$} & {$2_1$} & \red{$6_1$} & 
    \red{$0_0$} & {$5_2$} & {$4_2$} & {$7_3$}$)$
    \\
    $~$ & $~$ & $~$ & $~$ & {$~$} & {$~$} & {$~$} & $~$ &
    $~~$ & 
    $(${$1_0$} & {$3_0$} & \red{$2_1$} & \red{$0_0$} & 
    {$6_1$} & \red{$5_2$} & \red{$4_2$} & {$7_3$}$)$
    \\
    $~$ & $~$ & $~$ & $~$ & {$~$} & {$~$} & {$~$} & $~$ &
    $~~$ & 
    $(${$1_0$} & \red{$3_0$} & \red{$0_0$} & {$2_1$} & 
    {$6_1$} & {$4_2$} & {$5_2$} & {$7_3$}$)$
    \\
    $~$ & $~$ & $~$ & $~$ & {$~$} & {$~$} & {$~$} & $~$ &
    $~~$ & 
    $($\red{$1_0$} & \red{$0_0$} & {$3_0$} & {$2_1$} & 
    {$6_1$} & {$4_2$} & {$5_2$} & {$7_3$}$)$
    \\
    $~$ & $~$ & $~$ & $~$ & {$~$} & {$~$} & {$~$} & $~$ &
    $~~$ & 
    $(${$0_0$} & {$1_0$} & {$3_0$} & {$2_1$} & 
    {$6_1$} & {$4_2$} & {$5_2$} & {$7_3$}$)$
    \\ 
\end{tabular}
\end{center}

\end{example}

The more valid permutations of the set of transactions we have, more precisely we'll define the canonical form.     
We can distill the approach taken in the example into a proof:

\begin{proof}
    Let 
\[\begin{tikzcd}[ampersand replacement=\&,column sep=2.25em,row sep=tiny]
	{(q_0,u_0,t_0)} \& {(q_1,u_1,t_1)} \& \ldots \& {(q_n,u_n,t_n)} \& {s_1} \& {s_2} \& \ldots  \\
	{(q'_0,u'_0,t'_0)} \& {(q'_1,u'_1,t'_1)} \& \ldots \& {(q'_n,u'_n,t'_n)} \& {s_1'} \& {s_2'} \& \ldots
	\arrow[from=1-1, to=1-2]
	\arrow[from=1-2, to=1-3]
	\arrow[from=1-3, to=1-4]
	\arrow[from=1-4, to=1-5]
    \arrow[from=1-5, to=1-6]
    \arrow[from=1-6, to=1-7]
	\arrow[from=2-1, to=2-2]
	\arrow[from=2-2, to=2-3]
	\arrow[from=2-3, to=2-4]
    \arrow[from=2-4, to=2-5]
    \arrow[from=2-5, to=2-6]
    \arrow[from=2-6, to=2-7]
\end{tikzcd}
\vspace{-3mm}
\]
be two traces in the simple graph corresponding to $\LEDGER,$ where $u_0=u'_0\in\UTxO_0$ and $(t'_0,\ldots,t'_n)$ is a permutation of $(t_0,\ldots,t_n).$ Let us define sets $r_i = \fun{getORefs}(t_i)$ and $c_i = \fun{mkOuts}(t_i)$ for $0\leq i\leq n.$ Also, let $u_{\cup} = u_0 \cup_{0\leq i\leq n} c_i$ be the result of adding all outputs from all transactions $t_0, ... , t_n$ to $u_0$, which is also the same as adding all outputs from all transactions $t_0', ... , t_n'$ to $u_0$. The final states $u_n$ and $u_n'$, as well as all intermediate states, are contained in $u_{\cup}$, since the only other operation involved in updating the state is removing UTxO entries. 

Set $r = \cup_{0\leq i\leq n} r_i$ consists of UTxO entries that $t_0, ... , t_n$ (or $t_0', ... , t_n'$) remove from $u_{\cup}$. This $r$ is such that $r \subseteq u_{\cup}$, which is guaranteed by $\fun{checkTx}$. To show that both $u_n = u_{\cup} \setminus~r = u_n'$, we proceed by contradiction. Suppose $a \notin u_n$, but $a \in u_n'$. By above, $a \in u_{\cup}$. Now, $a$ must have been removed by some $t_i$ from $u_{\cup}$ (and therefore, also some $t_j'$). Since $a$ is in $u_n'$, it must have been re-added by another $t_k$. From the replay protection property, we get a contradiction: $t_k$ must have the same encoding as a previous transaction that added $a$, which contradicts Corollary \ref{cor:injective}. By similar logic, we can show that any $a \in u_n$ must also be in $u_n'$.
\end{proof}

The following algorithm produces a canonical form of a transaction list:

\begin{algorithm}[Canonical form of transaction sequence]
    
\textbf{Input data}. A finite path $\{(q_i,u_i,t_i)\}_{i=0}^{n}$ in a graph corresponding to $\LEDGER,$ where $u_0\in\UTxO_0.$ 

\textbf{Output data}. A set of all possible finite paths $\{(q'_i,u'_i,t'_i)\}_{j=0}^{n},$ where $u'_0 = u_0$ and $(t'_0,\ldots,t'_n)$ is a permutation of $(t_0,\ldots,t_n).$

\textbf{Step 1}. For every $i$ determine the set 
$
    K_i := 
    \left\{
        j
        ~|~
        r_i \cap c_j \neq \varnothing
    \right\}    .
$

\textbf{Step 2}. Define a poset structure on $\{0,\ldots,n\}:$ $i<j$ if and only if $j\in K_i.$

\textbf{Step 3}. Say that an index $i$ has level 0, if $K_i = \varnothing.$ If $K_i\neq \varnothing,$ say that a level of $i$ is a maximum of path lengths from $i$ to indexes of level 0 in Hasse diagram of the poset on $\{0,\ldots,n\}.$

\textbf{Step 4}. Make a total order on the set $\{0,\ldots,n\}:$ $i<j$ if and only if the level of $i$ is less that the level of $j,$ or $i<j$ in the natural ordering, if $i$ and $j$ have the same level. Call a sequence representing this total order
$
    (t_{i_0},\ldots,t_{i_n})    
$
the \textit{canonical presentation} of $(t_1,\ldots,t_n).$ 

\textbf{Step 5}. An \textit{elementary swap} of $(t_{j_0},\ldots,t_{j_k},\ldots,t_{j_l},\ldots,t_{j_n})$ is a sequence $(t_{j_0},\ldots,t_{j_l},\ldots,t_{j_k},\ldots,t_{j_n})$ if $l<k$ in a total order. Return as an output a set of all sequences
$
    (t_{j_0},\ldots,t_{j_n})
$
that can be obtained from the canonical presentation by a finite number of elementary swaps.
\end{algorithm}

Note that the decision procedure has the input of a single transaction list, while the transaction commutativity theorem is formulated in terms of two lists that are permutations of each other. It is more natural to express this theorem in a way that is symmetric with respect to the two permutations of transaction lists. To state it as a property, we can rephrase it as a property of a single trace: the property holds for a given trace with the required structure whenever, given any other trace with the required prefix and suffix, the final state of the prefix must be the same for both traces.

Like in the case of replay protection as well as trivial update protection, this is a safety property by virtue of being true for the entire space. This is also a safety property when considered as a property of all possible UTxO state traces (not just the valid ones). Suppose $s$ is a trace of UTxO states that has a $n$-length prefix that is not generated by applying a valid lists of transactions, and an arbitrary suffix. If another trace $s'$ exists such that $s$ and $s'$ violate the transaction commutativity property, changing the suffix starting at $n+1$ position in the trace of $s$ will never "fix" $s$. So, when the transaction commutativity property breaks in a finite prefix, it cannot be remedied, as required by the definition of a safety property.


\section{Conclusion}

\subsection{Related Work}
In this work, we use tools from different areas of mathematics to study program (i.e., 
smart contract) executions on the ledger. The approach to representing ledger semantics 
is used in existing work in UTxO ledger and contract formalization \cite{eutxo} \cite{alonzo} \cite{structured}.
Transition systems are commonly represented by graphs, with edges corresponding to possible state transitions \cite{milner}.
However, the graph generated directly by the small-steps semantics of ledgers (or other programs) 
does not contain edges for multi-step transitions, and it also does not exclude "bad" starting 
states. For this reason, we instead consider possible paths in the resulting graphs, and limit our 
attention to only paths with certain "good" starting states. 
Possible paths in our valid transition graphs align with the notion of execution traces \cite{liveness},
however, they are generated according to the small-steps specifications, and are therefore not arbitrary. 

The definition of ledger implementation relation, which we build on
this work, is a reversal of the the classic \emph{simulation} relation \cite{milner}. 
We introduce this terminology in this way because it is descriptive of the unique architecture of programs 
running on the UTxO ledger, which is different from the traditional communicating automata-style 
distributed programs. This is because stateful UTxO programs are implemented using multiple 
permission-like pieces of code that control what aggregates (or pieces) of the ledger state that
transactions are permitted to control. It also does not make sense to model such programs via  
a notion that is used in studying transition systems --- subsystems \cite{universal}. This is because 
a subsystem relates the evolution of a subset of states to the evolution of the entire system, 
which is not what we study here. 

Rather than subsystems, we use sieve-defined homomorphisms as the formalism 
that relates the evolution of the implemented program state to the that of the ledger state.
In category theory, a sieve is a is a generalization of the notion of ideals, which guarantees 
that arrows starting in a given set of objects also end in that same set \cite{sieves}.
Now, categories are closed under arrow composition, and our graphs are not closed under 
edge composition. We, however, apply the defining property of categorical sieves to graph
homomorphisms.

Algebraic descriptions of UTxO transaction processing appear in existing work \cite{blockalg}. 
However, no models of
implementations of programs on the ledger are described within this model. The ledger structure itself has been 
modelled categorically, using monoidal functors to represent the structure of UTxO set updates. 
It may be possible to combine this model and ours in future work. 

Certain safety properties have been defined, and are consistently being checked in
production ledger systems as part of property testing
frameworks. This involves generating arbitrary transactions, and applying them to 
generated states to verify that the property is not violated \footnote{%
\url{https://github.com/IntersectMBO/cardano-ledger}}. In the future, we also plan 
to incorporate the results in this work into more realistic models. 

\subsection{Discussion}
Programs for UTxO-style ledgers, which are composed entirely out of predicates on state updates,
differs significantly from common programming paradigms. For this reason, a specialized model 
is required to describe it in a principled way.
The structured contract framework provides a rigorous and principled way to establish a relation between
the contents of pieces of user-defined code, including a small-steps specification, an implementation, 
and a proof that a single valid ledger step corresponds to a valid program step. However, 
the capacity for reasoning within this framework is very limited without additional 
structure. In this work, we define the required structure. 

Our contribution includes a rigorous definition of a valid program trace for a given 
small-step semantics and set of valid initial states. 
Next, we construct a category that we use to model the relation between 
the behavior of the ledger and a program implemented on it. This category 
has simple graphs as objects, and partial sieve-defined homomorphisms as morphisms. We define
maps between valid traces, which are paths in the graph, in terms of maps in $\GraphSieve_{*}$.
We then show how such maps are induced by a structured contract. 
We apply an existing metric defined on arbitrary execution traces (i.e., paths in 
a complete directed state graph) to 
paths in our graph, which correspond to valid traces only. We show that the metric we applied 
is, in fact, an ultrametric. This allows us to demonstrate that 
the maps in our category are non-expanding and continuous, and therefore, for any 
safety property of a program, its ledger implementation has an associated safety property. 

We go on to prove certain important safety properties of the ledger, including commutativity, 
replay protection, and trivial update protection. These properties
all rely on the ability to reason about the prefix of a given trace, which was previously 
not supported in the structured contract framework. 
Assuming certain consequences of these properties
is often required in practice for proving correctness of program implementations on the ledger \cite{msgs}.

We formalized the notion of safety in the UTxO ledger programming context, leaving liveness for future work. 
We intend to use our model to study relationships between structured contracts
on a single ledger,
as well as the possibility of composing them. We are also interested in
investigating unique quirks of smart contracts in the EUTxO model, such as the double satisfaction
problem \cite{msgs}, and formalize what it means for a transaction's
changes to the ledger to be predictable (building on the transaction commutativity
property).
Another interesting direction of research would be to apply similar methods and techniques 
we have defined here to other blockchains that have formal models, such as Algorand \cite{algorandformal}
or Bitcoin \cite{bitcoinformal}.

\nocite{*}
\bibliographystyle{eptcs}
\bibliography{sources}
\end{document}